\documentclass[a4paper,11pt]{article}

\usepackage{geometry}
\usepackage[latin1]{inputenc}
\usepackage{amsfonts}
\usepackage{amsmath}
\usepackage{amssymb}
\usepackage{amsthm}
\usepackage[T1]{fontenc}
\usepackage{bbm}

\theoremstyle{plain}
\newtheorem{theorem}{Theorem}
\newtheorem{corollary}[theorem]{Corollary}
\newtheorem{proposition}[theorem]{Proposition}

\theoremstyle{definition}

\newtheorem{condition}[theorem]{Condition}
\newtheorem{observation}[theorem]{Observation}

\theoremstyle{remark}

\newcommand{\wick}[1]{\text{:}#1\text{:}}
\newcommand{\st}[1]{\mathbf{#1}}

\begin{document}

\title{Local Thermal Equilibrium and KMS states\\ in Curved Spacetime}
\author{Christoph Solveen\\\quad\\ \small Institut f\"ur Theoretische Physik, Universit\"at G\"ottingen,\\ \small Friedrich-Hund-Platz 1, 37077 G\"ottingen, Germany\\
\quad\\ \small Christoph.Solveen@theorie.physik.uni-goettingen.de
\quad}
\maketitle

\begin{abstract}
On the example of a free massless and conformally coupled scalar field, it is argued that in quantum field theory in curved spacetimes with time-like Killing field, the corresponding KMS states (generalized Gibbs ensembles) at parameter $\beta>0$ need not possess a definite temperature in the sense of the zeroth law. In fact, these states, although passive in the sense of the second law, are not always in local thermal equilibrium (LTE). A criterion characterizing LTE states with sharp local temperature is discussed. Moreover, a proposal is made for fixing the renormalization freedom of composite fields which serve as ``thermal observables'' and a new definition of the thermal energy of LTE states is introduced. Based on these results a general relation between the local temperature and the parameter $\beta$ is established for KMS states in (Anti) de Sitter spacetime.
\end{abstract}

\renewcommand\abstractname{Mathematics Subject Classification (2010)}

\begin{abstract}
  81T20, 81T28, 82C99
\end{abstract}

\section{Introduction}

In quantum field theory (QFT) in curved spacetime, the determination of thermal parameters such as temperature is made difficult by the presence of tidal forces. Due to the general lack of global time-like isometries there is no acceptable notion of global thermal equilibrium in the presence of curvature. However, in the special case of spacetimes with time-like Killing field, the corresponding isometries induce a time evolution for quantum fields propagating on such stationary backgrounds \cite{Full89,Wald94}. The KMS condition then distinguishes states which are stationary, stable and passive with respect to this dynamics at some parameter $\beta>0$. This condition is the suitable replacement for Gibbs ensembles in the case of infinitely extended media \cite{Haag96}. Passivity is meant in the sense of the second law: an observer following the flow lines of the Killing field cannot extract energy from passive states by running a cyclic engine.

In Minkowski spacetime the KMS parameter $\beta$ has the meaning of an inverse temperature.\footnote{We use units where $k_{B}=1$.} Recall that according to the zeroth law temperature is a parameter distinguishing equilibrium states which are able to coexist if brought in thermal contact. The second law then entails a canonical choice, the ``absolute'' temperature $T$, given by $T=1/\beta$.

In this letter, we argue that in curved spacetime, $\beta$ is in general not directly related to temperature in the sense of the zeroth law. For in stationary curved spacetimes, tidal forces and the Unruh effect impede this interpretation of the KMS parameter. One can, however, ask whether KMS states in curved stationary spacetimes are locally thermal and, if so, how $\beta$ is related to the \emph{local} temperature.

It was argued in \cite{BuchSchl07,SchlVerc08,Solv10} that, in curved spacetime, local concepts may replace the global ones involved in the definition of equilibrium. In locally covariant QFT \cite{BrunFredVerc03}, one can use point-like localized observables which are sensitive to thermal properties of global equilibrium states in Minkowski spacetime in order to attach local thermal parameters like a temperature to ``local thermal equilibrium'' (LTE) states in curved spacetime. In this approach, LTE states are distinguished by the fact that they respond to the point-like ``thermal observables'' in the same way as global equilibrium states. This notion of local equilibrium leads to a generalization of the Unruh effect \cite{BuchSchl07}, opens new perspectives on the problem of the ``arrow of time'' \cite{Buch03} and could also be of relevance in cosmology \cite{Verch11}.

\section{LTE states for the scalar field}

For the definition of LTE states, recall that a locally covariant quantum field $\varphi$ is defined on all suitable spacetimes $\st{M}$ simultaneously.\footnote{Here, the notation $\st{M}$ subsumes the relevant geometric structure of the spacetime, i.e.\ a metric, an orientation and a time orientation. We use metric signature $(-,+,+,+)$.} For each $\st{M}$,
\begin{displaymath}
 \varphi_{{}_{\st{M}}}(x),~~~~~~~x\in\st{M},
\end{displaymath}
is a quantum field on $\st{M}$ which transforms covariantly under isometric embeddings. On each $\st{M}$, the set of states consists of expectation functionals on the polynomials of the fields.

Let us assume for simplicity that in Minkowski spacetime $\st{Mink}$, for each Lorentz system the global equilibrium state, denoted by $\omega_{\beta}$, is unique for given $\beta>0$. The $\omega_{\beta}$ are then homogeneous. According to \cite{BuchSchl07}, given a set $\mathbf{T}$ of thermal observables, a state $\omega$ is ``$\mathbf{T}$-thermal at $x\in\st{M}$'' (or ``$\mathbf{T}(x)$-thermal''), if there exists a suitable normalized measure $\rho_{x}$ such that
\begin{equation}\label{LTE}
\omega(\phi_{{}_{\st{M}}}(x))=\int\,d\rho_{x}(\beta)\,\omega_{\beta}(\phi_{{}_{\st{Mink}}}(0))
\end{equation}
for all $\phi\in\mathbf{T}$. In other words, a state is $\mathbf{T}$-thermal at $x$ if it responds like a (mixture of) global equilibrium state(s) when tested with fields in $\mathbf{T}$ at $x$. It follows that the expectation value of the macroscopic observable $\Phi$,
\begin{equation}\label{LTE2}
\Phi(\beta):=\omega_{\beta}(\phi_{{}_{\st{Mink}}}(0)),
\end{equation}
can consistently be defined in the LTE state $\omega$. Concretely, the expectation value of $\Phi$ in $\omega$ at spacetime point $x$ is defined as
\begin{equation}\label{ExpThermFunct}
\omega(\Phi)(x):=\omega(\phi_{{}_{\st{M}}}(x))\stackrel{\eqref{LTE}}{=}\int\,d\rho_{x}(\beta)\,\Phi(\beta).
\end{equation}
See \cite{BuchOjimRoos02,BuchSchl07} and \cite{Solv10} for further discussion.

We discuss here the simple example of the massless, conformally coupled free scalar field $\phi$ with equation of motion
\begin{equation}\label{eom}
P\phi_{{}_{\st{M}}}:=\left(-\Box_{{}_{\st{M}}}+\frac{1}{6}R_{{}_{\st{M}}}\right)\phi_{{}_{\st{M}}}=0,
\end{equation}
where $\Box_{{}_{\st{M}}}$ and $R_{{}_{\st{M}}}$ denote the D'Alembertian and the Ricci scalar on $\st{M}$ respectively. The canonical commutation relations for the basic field are given by
\begin{displaymath}
 [\phi_{{}_{\st{M}}}(x),\phi_{{}_{\st{M}}}(y)]=i\Delta_{{}_{\st{M}}}(x,y)\mathbbm{1}
\end{displaymath}
with the commutator function $\Delta_{{}_{\st{M}}}$, given by the difference of the retarded and advanced fundamental solutions of \eqref{eom}. We restrict attention here to the family of Hadamard states, see e.g.\ \cite{Sand10} and references therein.

The thermal observables for the free field are chosen among the locally covariant Wick polynomials, which are unique up to renormalization constants \cite{HollWald01}. The choice of these numbers is therefore an important part of the definition of LTE. For example, the locally covariant Wick square is defined by point splitting, i.e.\
\begin{equation}\label{WickConstant}
\wick{\phi^2}_{{}_{\st{M}}}(x):=\lim_{y\rightarrow x}\,\Bigl(\phi_{{}_{\st{M}}}(x)\phi_{{}_{\st{M}}}(y)-h(x,y)\Bigr)+\alpha\, R_{{}_{\st{M}}}(x),
\end{equation}
where $h$ is the Hadamard parametrix, defined in a neighbourhood of $(x,x)$, and $\alpha\in\mathbb{R}$ is an a priori undetermined renormalization constant. The coincidence limit on the right hand side is the Hadamard normal ordered Wick square, denoted by $\wick{\phi^2}_{h}(x)$.\footnote{Higher Wick powers and their derivatives are defined analogously and exhibit similar renormalization ambiguities, an important example being the stress energy tensor.} In global equilibrium states in Minkowski spacetime,
\begin{equation}\label{TempMink}
 \omega_{\beta}(\wick{\phi^2}_{{}_{\st{Mink}}}(x))=\frac{1}{12\beta^{2}}.
\end{equation}
In view of relation \eqref{LTE2}, $\wick{\phi^2}$ may thus be used as a local thermometer in arbitrary spacetimes and any state with $\omega(\wick{\phi^2}_{{}_{\st{M}}}(x))\geq 0$ is locally thermal at $x$ with regard to this observable. According to equations \eqref{ExpThermFunct} and \eqref{TempMink}, the mean of the square of the local temperature in this state is given by
\begin{equation}\label{LocalTemp}
\overline{T^2(x)}:=\int\,d\rho_{x}(\beta)\frac{1}{\beta^{2}}=12\,\omega(\wick{\phi^2}_{{}_{\st{M}}}(x))
\end{equation}
for some measure $\rho_{x}$.

Usually, there exist many measures such that equation \eqref{LocalTemp} holds, so one needs further thermal observables in order to determine the distribution of the thermal functions. Here, we sketch a way to decide whether a given LTE state $\omega$ has a sharp local temperature $T(x)$, i.e.\ if $\rho_{x}$ is a Dirac measure. Moreover, we show that there are always states of sharp local temperature in compact regions of any curved spacetime. To do so, we choose the set of thermal observables as
\begin{equation}
 \mathbf{T}_{2n}:=\{\mathbbm{1},\wick{\phi^2},\dots,\wick{\phi^{2n}}\}
\end{equation}
for some integer $n\geq0$. According to \cite[Theorem 5.1]{HollWald01}, the renormalization freedom for the Wick powers is given by
\begin{equation}\label{RenormWick}
\wick{\phi^{n}}_{{}_{\st{M}}}(x)=\wick{\phi^{n}}_{h}(x)+\sum^{n-2}_{k=0}\binom{n}{k}C_{n-k}(x)\,\wick{\phi^{k}}_{h}(x).
\end{equation}
Here, the $C_{i}(x)$ are polynomials with real coefficients in the metric and curvature with suitable scaling behaviour, e.g.\ $C_{2}=\alpha R_{{}_{\st{M}}}$, see \eqref{WickConstant}.

Recall that a state $\omega$ is called quasi-free if the corresponding $n$-point functions vanish for odd $n$ and if
\begin{align*}
 \omega(\phi_{{}_{\st{M}}}(x_{1}),\dots,\phi_{{}_{\st{M}}}(x_{2n}))=\sum_{\pi\in\Pi_{n}}&\omega(\phi_{{}_{\st{M}}}(x_{\pi(1)}),\phi_{{}_{\st{M}}}(x_{\pi(2)}))\dots\\
 \dots\,&\omega(\phi_{{}_{\st{M}}}(x_{\pi(2n-1)}),\phi_{{}_{\st{M}}}(x_{\pi(2n)})),
\end{align*}
where $\Pi_{n}$ is the set of permutations of $\{1,\dots,2n\}$ with $\pi(1)<\pi(3)<\dots<\pi(2n-1)$ and $\pi(2i-1)<\pi(2i)$, $i=1,\dots,n$. In quasi-free states one has
\begin{equation}\label{CombQuasifree}
 \omega\bigl(\wick{\phi^{2n}}_{h}(x)\bigr)=(2n-1)!!\,\omega\bigl(\wick{\phi^{2}}_{h}(x)\bigr)^n
\end{equation}
for the even Hadamard normal ordered Wick powers, while the odd ones have vanishing expectation values. Since the global equilibrium states are quasi-free and the $C_{i}$ are zero in Minkowski spacetime, one finds
\begin{equation}\label{WickThermFunct}
 \omega_{\beta}\bigl(\wick{\phi^{2n}}_{{}_{\st{Mink}}}(0)\bigr)=(2n-1)!!\,\omega_{\beta}\bigl(\wick{\phi^{2}}_{{}_{\st{Mink}}}(0)\bigr)^n=\frac{(2n-1)!!}{12^n\,\beta^{2n}}.
\end{equation}
\begin{proposition}\label{QuasiMink}
 A quasi-free state in Minkowski spacetime is $\mathbf{T}_{2n}(x)$-thermal for all $n$ if and only if it is $\mathbf{T}_{2}(x)$-thermal. As a consequence, a $\mathbf{T}_{2}(x)$-thermal quasi-free state in Minkowski spacetime has a sharp local temperature.
\end{proposition}
\begin{proof}
The first observation is a direct consequence of relation \eqref{WickThermFunct}. The second statement can be justified as follows:\ in a quasi-free state
\begin{equation}
 \omega(\wick{\phi^4}_{{}_{\st{Mink}}}(x))=3\,\omega(\wick{\phi^2}_{{}_{\st{Mink}}}(x))^2.
\end{equation}
But the $\mathbf{T}_{2}(x)$-thermal state $\omega$ is also $\mathbf{T}_{4}(x)$-thermal and thus there exists some measure $\rho_{x}$ such that
\begin{equation*}
 \int \frac{d\rho_{x}(\beta)}{(\beta^2)^2}=\left(\int \frac{d\rho_{x}(\beta)}{\beta^2}\right)^2.
\end{equation*}
This equation implies that the Cauchy-Schwarz inequality for the functions $1$ and $1/\beta^2$ is a sharp equality on $L^2(\mathbb{R}^{+},\rho_{x})$, i.e.\ $1/\beta^2=\text{const.}\cdot 1$ almost everywhere with respect to the measure $\rho_{x}$. By the monotonicity of $\beta\mapsto 1/\beta^2$ this means that $\rho_{x}$ is supported at a single point - it must be a Dirac measure.
\end{proof}
This conclusion rests on the combinatorics of Hadamard normal ordered Wick powers in quasi-free states. However, for the full Wick powers in curved spacetime (including the renormalization freedom) we have
\begin{equation}
 \omega\bigl(\wick{\phi^{2n}}_{{}_{\st{M}}}(x)\bigr) = \omega\bigl(\wick{\phi^{2n}}_{h}(x)\bigr)+ \sum^{2n}_{k=1}\binom{2n}{k}\,C_{2n-k}(x)\,\omega\bigl(\wick{\phi^{k}}_{h}(x)\bigr),
\end{equation}
according to equation \eqref{RenormWick}. Imposing relation \eqref{CombQuasifree} for the full locally covariant Wick powers up to order $2n$ therefore greatly reduces the renormalization freedom.
\begin{proposition}\label{PropQuasifreeComb}
Let $\omega$ be a quasi-free state on some spacetime $\st{M}$. Then
\begin{equation}\label{CombCurved}
 \omega\bigl(\wick{\phi^{2n}}_{{}_{\st{M}}}(x)\bigr)=(2n-1)!!\,\omega\bigl(\wick{\phi^{2}}_{{}_{\st{M}}}(x)\bigr)^n
\end{equation}
if and only if
\begin{equation}\label{CondRenorm}
 C_{2m}=(2m-1)!!\,C_{2}{}^{m}~~~\text{for all}~m=1,2,\dots,n.
\end{equation}
\end{proposition}

\begin{proof}
The straightforward combinatorial argument can be found in \cite{Solv12}.
\end{proof}

Note that condition \eqref{CondRenorm} is non-trivial. For example $C_{4}$ is of the form
\begin{displaymath}
 C_{4}=\gamma_{0}R^2+\gamma_{1}\Box R+\gamma_{2} R_{ab}R^{ab}+\gamma_{3}R_{abcd}R^{abcd}
\end{displaymath}
with real constants $\gamma_{i}$. Remembering that $C_{2}=\alpha\,R$, the proposition states that if \eqref{CombCurved} is to hold for quasi-free states in arbitrary spacetimes $\st{M}$, then $\gamma_{1}=\gamma_{2}=\gamma_{3}=0$ and $\gamma_{0}=3\alpha^2$. More generally, \eqref{CondRenorm} means
\begin{equation}\label{CR}
 C_{2m}=(2m-1)!!\,\alpha^{m}\,R^{m}~~\text{for all}~~m=1,2,\dots,n.
\end{equation}
Once the Wick square is fixed, condition \eqref{CombCurved} therefore completely fixes the even Wick powers up to order $2n$.
\begin{corollary}\label{CorTherm}
If \eqref{CR} holds for all $m=1,2,\dots,n$, then a quasi-free state is $\mathbf{T}_{2n}(x)$-thermal if and only if it is $\mathbf{T}_{2}(x)$-thermal. If $n\geq 2$, then the state has a sharp local temperature.
\end{corollary}
Adopting \eqref{CR} for some $n\geq 2$ as our renormalization prescription, we conclude that all quasi-free LTE states necessarily have a sharp local temperature. The KMS states in the examples discussed below are all of this type.

\begin{proposition}
Let $\mathcal{O}$ be a compact subset of some spacetime $\st{M}$ and let $n\in\mathbb{N}$. If \eqref{CR} holds for all $m=1,2,\dots,n$, then there exist (not necessarily quasi-free) $\mathbf{T}_{2n}(\mathcal{O})$-thermal states.
\end{proposition}

\begin{proof}
A state is $\mathbf{T}_{2}(x)$-thermal if the corresponding expectation value of $\wick{\phi^2}_{{}_{\st{M}}}(x)$ is positive. We can construct a state with this property for all $x\in\mathcal{O}$ as follows. Pick any state $\omega$ with non-vanishing one-point function $\omega^{(1)}$ and modify its two-point function $\omega^{(2)}$ by adding a term
\begin{equation*}
 \tilde{\omega}^{(2)}_{c}:=\omega^{(2)}+c~\omega^{(1)}\otimes\omega^{(1)}
\end{equation*}
with some $c\geq0$.\footnote{In fact, instead of $\omega^{(1)}$ we could take any smooth solution of $P$.} Since $\omega$ is Hadamard, $\omega^{(1)}$ is smooth. $\tilde{\omega}^{(2)}_{c}$ is a valid two-point function and we can consider the corresponding quasi-free state $\tilde{\omega}_{c}$.\footnote{The fact that a positive bi-distribution gives rise to a quasi-free state is proven in \cite{Petz90} at the level of Weyl algebras.}
Clearly
\begin{equation*}
 \tilde{\omega}_{c}(\wick{\phi^2}_{{}_{\st{M}}}(x))=\omega(\wick{\phi^2}_{{}_{\st{M}}}(x))+c\,(\omega^{(1)}(x))^2.
\end{equation*}
Due to the compactness of $\mathcal{O}$, we may choose the constant $c$ such that $\tilde{\omega}_{c}(\wick{\phi^2}_{{}_{\st{M}}}(x))$ is positive for all $x\in\mathcal{O}$. The corresponding states $\tilde{\omega}_{c}$ are consequently $\mathbf{T}_{2}(\mathcal{O})$-thermal. By Corollary \ref{CorTherm} this means that these states, being quasi-free, are also $\mathbf{T}_{2n}(\mathcal{O})$-thermal. Now we use the fact that $\mathbf{T}_{2n}(\mathcal{O})$-thermality is stable under convex combinations. We simply mix the so-constructed quasi-free states and thus find states which are not quasi-free but still $\mathbf{T}_{2n}(\mathcal{O})$-thermal.
\end{proof}
This result shows that even if curved spacetime KMS states fail to be in local equilibrium, there are still many LTE states (even with sharp local temperature) in any compact region of curved spacetime. This can be seen as a consistency check on our definition of LTE states.

\section{KMS states and renormalization freedom}

From equations \eqref{WickConstant} and \eqref{LocalTemp} it is clear that $\mathbf{T}_{2n}$-thermality and the value of the local temperature depend on the choice of the renormalization constant $\alpha$. Following an idea by Buchholz and Schlemmer \cite{BuchSchl07}, we fix $\alpha$ based on the investigation of KMS states with respect to Killing flows in Anti de Sitter, Rindler, static de Sitter and Schwarzschild spacetimes. We also explain in which sense our proposal differs from the one presented in \cite{BuchSchl07}. Moreover, with our choice, there is an interesting relation between the KMS parameter $\beta$ and the local temperature, generalizing earlier results on the Hawking Unruh temperature in \cite{NarPetThi96}, \cite{DesLev97} and \cite{Jac98}.

Anti de Sitter spacetime $\st{AdS}$ is the maximally symmetric spacetime with constant negative curvature. It can be represented as the universal cover of
\begin{equation}
 \{x\in\mathbb{R}^5~|~-(x^{0})^2-(x^{1})^2+(x^{2})^2+(x^{3})^2+(x^{4})^2=-\rho^2\}
\end{equation}
with radius $\rho$. We assume that the free scalar quantum field and its Wick powers have been constructed with ``transparent'' boundary conditions for the Green's functions of the field, see \cite{AviIshSto78}.

The trajectories of the Killing field $K:=x^{2}\partial_{0}-x^{0}\partial_{2}$ are positive time-like only in the interior of the wedge
 \begin{displaymath}
 W^{+}:=\{ x\in\st{AdS} ~|~ |x_{0}|< x_{2} \}.
\end{displaymath}
For each $\beta>0$, an observer following the flow lines of $K$ can prepare quasi-free KMS states $\omega_{\beta}$ in the interior of $W^{+}$. For the Wick square, one finds\footnote{This result is new. However, we do not put the calculation on record here, because there are no new ideas involved. The methods used are those developed by Stottmeister in \cite{Stot09}.}
\begin{equation}\label{WickAdS}
 \omega_{\beta}(\wick{\phi^2}_{{}_{\st{AdS}}}(x))=\frac{1}{-\langle K,K\rangle_{x}}\left(\frac{1}{12\beta^2}-\frac{1}{12(2\pi)^2}\right)+\frac{1}{24\pi^2\rho^2}-\frac{12}{\rho^2}\alpha.
\end{equation}
The Hawking-Unruh type state with $\beta=2\pi$ is distinguished by the fact that it may be extended to all of $\st{AdS}$. Moreover, it is a ground state with respect to the global dynamics induced by the time-like Killing field $x^{0}\partial_{1}+x^{1}\partial_{0}$, see \cite{BuchFlorSumm00}. We now propose to fix the renormalization constant as
\begin{equation}\label{choice}
 \alpha=\frac{1}{288\pi^2}.
\end{equation}
With this choice, \eqref{WickAdS} becomes
\begin{equation}\label{WickAdS2}
 \omega_{\beta}(\wick{\phi^2}_{{}_{\st{AdS}}}(x))=\frac{1}{-\langle K,K\rangle_{x}}\left(\frac{1}{12\beta^2}-\frac{1}{12(2\pi)^2}\right).
\end{equation}
As a consequence, the Hawking-Unruh type state has a local temperature of zero everywhere. Moreover, KMS states with $\beta>2\pi$ have negative expectation value for the Wick square and are therefore out of local equilibrium.

In terms of the squared four acceleration $a^2(x)$, one has
\begin{equation}
 \frac{1}{-\langle K,K\rangle_{x}}=a^2(x)-\frac{1}{\rho^2}.
\end{equation}
Note that the infimum of $a^2(x)$ along the flow lines of $K$ in $W^{+}$ is given by $\frac{1}{\rho^2}$.

Generally, an observer following the flow lines of a Killing field  $K$ through a point $x$ perceives the KMS parameter to be scaled, because his proper time is scaled with regard to the Killing parameter. Concretely:\footnote{Recall that the quantity $\langle K,K\rangle_{x}$ is constant along the flow lines of $K$.}
\begin{equation}
\beta_{{}_{\mathrm{KMS}}}(x):=\sqrt{-\langle K,K\rangle_{x}}\,\beta.
\end{equation}
In the case of the $\st{AdS}$ dynamics discussed here, one finds that $-\langle K,K\rangle_{x}\rightarrow\infty$ if $x$ tends to spatial infinity.\footnote{Of course, the trajectories are the ones with $x_{2}^2-x_{0}^2=\text{const.}$ and by ``spatial infinity'' we mean this constant tending to $+\infty$.} Hence $1/\beta_{{}_{\mathrm{KMS}}}(x)\rightarrow 0$ at spatial infinity. Therefore, for the Hawking-Unruh type state, $1/\beta_{{}_{\mathrm{KMS}}}(x)$ tends to the value of the local temperature (zero, in this case) for trajectories far away from the origin.

In Rindler spacetime $\st{R}:=\{x\in \mathbb{R}^4~|~~|x^0|<x^1\}$ we find an analogous situation, which however is not influenced by the choice of $\alpha$, since $\st{R}$ is flat. There is a positive time-like Killing vector $K=x_{1}\partial_{0}-x_{0}\partial_{1}$, which generates the Lorentz boosts in the $x^1$-direction. Again, an observer following this flow can prepare quasi-free KMS states for any value $\beta>0$. One finds \cite{Pete07}
\begin{equation}\label{WickRindler}
 \omega_{\beta}(\wick{\phi^2}_{{}_{\st{R}}}(x))=\frac{1}{-\langle K,K\rangle_{x}}\left(\frac{1}{12\beta^2}-\frac{1}{12(2\pi)^2}\right).
\end{equation}
Here, the four acceleration is related to the norm of the Killing field by
\begin{displaymath}
 \frac{1}{-\langle K,K\rangle_{x}}=a^2(x).
\end{displaymath}
The Unruh state with $\beta=2\pi$ has vanishing expectation value, i.e.\ a local temperature of zero and KMS states with $\beta>2\pi$ are not in local equilibrium. As in Anti de Sitter, $-\langle K,K\rangle_{x}\rightarrow\infty$ for $x$ tending towards spatial infinity. Hence for the Unruh state the value of $1/\beta_{{}_{\mathrm{KMS}}}(x)$ tends to the local temperature (namely zero), just like in $\st{AdS}$.

We also consider exterior Schwarzschild spacetime $\st{Schw}$ surrounding a static black hole of mass $M$. It is a vacuum solution\footnote{This implies that there is no renormalization ambiguity for the Wick square, since $R_{{}_{\st{Schw}}}=0$.} to Einstein's equations with positive time-like Killing field $K$. In terms of the usual radial coordinate $r\in(2M,\infty)$,
\begin{equation}
-\langle K,K\rangle_{r}=1-\frac{2M}{r}.
\end{equation} The Hartle Hawking state is a quasi-free KMS state at $\beta=8\pi M$ with respect to the corresponding dynamics \cite{Wald94}. According to \cite{Page82}, in this state the Wick square takes the approximate value 
\begin{equation}\label{WickSchw}
 \omega_{8M\pi}(\wick{\phi^2}_{{}_{\st{Schw}}}(r))\approx\frac{1}{12(8\pi M)^2}\left(1+\frac{2M}{r}+\left(\frac{2M}{r}\right)^2+\left(\frac{2M}{r}\right)^3\right).
\end{equation}
This quantity approaches $\frac{1}{12(8\pi M)^2}$ as $r\rightarrow\infty$. Therefore, the inverse local temperature of the Hartle Hawking state equals $8\pi M$ at spatial infinity. We see that since $-\langle K,K\rangle_{r}\lim_{r\rightarrow\infty}=1$, the local temperature agrees with the scaled inverse KMS parameter at large radii. This resembles the situation in the previous examples.
\begin{observation}\label{Obs}
In Anti de Sitter, Rindler and Schwarzschild spacetime, for the Hawking-Unruh type states the inverse scaled KMS parameter and the local temperature agree at spatial infinity. 
\end{observation}
Recall that for $\st{R}$ and $\st{Schw}$, these findings are not influenced by the renormalization constant $\alpha$. However, it is our choice of $\alpha$ which entails that the situation in $\st{AdS}$ is similar to these examples.

We finally consider de Sitter spacetime $\st{dS}$, which can be represented as the hyperboloid
\begin{equation*}
 \{x\in\mathbb{R}^5~|~-x_{0}^2+x_{1}^2+x_{2}^2+x_{3}^2+x_{4}^2=\rho^2\}
\end{equation*}
with radius $\rho$. The Killing vector field $K:=x_{1}\partial_{0}-x_{0}\partial_{1}$ is positive time-like only in the wedge
\begin{displaymath}
 W^{+}=\{ x\in\st{dS}~|~~ |x_{0}|< x_{1} \}.
\end{displaymath}
$ W^{+}$ is the union of the past and future of the time-like geodesic
\begin{equation}\label{GeoDeSi}
\gamma(t):=(\rho\sinh(t),\rho\cosh(t),0,0,0),
\end{equation}
which is among the flow lines of $K$. As a globally hyperbolic spacetime, the interior of the wedge is called static de Sitter spacetime. Note that there is no spatial infinity for the observers on flow lines of $K$, so there is no analogue to Observation \ref{Obs} in this case.

But still, for $\beta>0$ one may construct quasi-free KMS states $\omega_{\beta}$ for the dynamics induced by $K$, and Buchholz and Schlemmer have calculated the expectation value of the Wick square in these states \cite{BuchSchl07,Schl05}.\footnote{As pointed out by Stottmeister \cite{Stot09}, there appears an error in \cite{BuchSchl07} in the derivation of this result. See \cite{Solv12} for details.}
With our choice of $\alpha$, it reads
\begin{equation}\label{WickdS}
 \omega_{\beta}(\wick{\phi^2}_{{}_{\st{dS}}}(x))=\frac{1}{-\langle K,K\rangle_{x}}\left(\frac{1}{12\beta^2}-\frac{1}{12(2\pi)^2}\right),
\end{equation}
showing that the Gibbons-Hawking state ($\beta=2\pi$) has a local temperature of zero. Note that the latter state is the only KMS state which can be extended to a full Hadamard state on all of de Sitter spacetime \cite{NarPetThi96}. In $\st{dS}$, the four acceleration is given by
\begin{equation}
\frac{1}{-\langle K,K\rangle_{x}}=a^2(x)+\frac{1}{\rho^2}.
\end{equation}

We can now use what we have learned in the previous examples. As a consequence of our choice of $\alpha$, there is an interesting relation between the inverse scaled KMS parameter $1/\beta_{{}_\mathrm{KMS}}(x)$ and the local temperature $T(x)$ in the wedges of Rindler, de Sitter and Anti de Sitter spacetime. In fact, it is easy to see that
\begin{equation}\label{KMSvsLocal}
 \frac{1}{\beta^2_{{}_\mathrm{KMS}}(x)}=\pm\frac{1}{(2\pi\rho)^2}+\frac{a^2(x)}{(2\pi)^2}+T^2(x),
\end{equation}
where the $+$ sign holds for de Sitter, the $-$ sign holds for Anti de Sitter and for Rindler spacetime, the first term on the right hand side is absent.

We see that there is a geometric contribution due to spacetime curvature corresponding to the first term on the RHS of \eqref{KMSvsLocal}, while the effect of the acceleration is encoded in the second one (i.e.\ the Unruh effect). These two contributions combine with the local temperature of the state, represented by the third term in \eqref{KMSvsLocal}, so that the state satisfies the KMS condition at parameter $\beta_{{}_\mathrm{KMS}}(x)$. Thus the passivity of these states is not only thermal in origin, but also stems from other effects, which are embodied in equation \eqref{KMSvsLocal}. This represents a generalization of formulas for the Hawking-Unruh type states discussed in \cite{NarPetThi96,DesLev97,Jac98}.

We thus propose the following prescription for fixing the renormalization freedom of the thermal observables.
\begin{condition}\label{CondRenKMS}
The renormalization constants in composite fields serving as thermal observables must be chosen according to the following condition. In Hawking-Unruh type KMS states in curved spacetime the scaled inverse KMS parameter has to approach the local temperature obtained with these fields on Killing trajectories tending to space-like infinity.
\end{condition}
In the massless, conformally coupled case, this choice entails the conformal covariance of the even Wick powers in the sense of \cite{Pin09}. More precisely, our choice makes the Wick square a conformally covariant quantum field. Since we also adopt condition \eqref{CombCurved}, all even Wick powers are also conformally covariant, which can be deduced from \cite[Theorem 3.1]{Pin09}.\footnote{Note that our value of the renormalization constant differs from the one given in \cite{Pin09}, due to an error in that reference. This is explained in more detail in \cite{Solv12}.} We note that with our choice for $\alpha$, the Gibbons Hawking state in de Sitter, the Unruh Hawking type state in Anti de Sitter and the Unruh state in Rindler spacetime all have a local temperature of zero. This is not surprising, since all of these states are conformal vacua in their respective spacetimes \cite{BirDav84}. But the condition should also be applicable in massive theories.

Buchholz and Schlemmer \cite{BuchSchl07} have proposed an $\alpha$ such that, in static de Sitter spacetime, the ground state $\beta\rightarrow\infty$ has local temperature zero on the geodesic $\gamma(t)$ in equation \eqref{GeoDeSi}.\footnote{On the geodesic, the inverse local temperature then agrees with the KMS parameter for all KMS states.} The motivation behind this choice is that, while the time-like Killing trajectories in $\st{dS}$ are geometrically quite unlike time translations in Minkowski spacetime, at least along the geodesic the KMS states are supposed to behave like equilibrium states when tested with scalar observables.

The argument of Buchholz and Schlemmer only works in de Sitter spacetime; in fact, it can be seen that with their choice, the Hawking-Unruh type state in Anti de Sitter spacetime is not in LTE. We therefore propose Condition \ref{CondRenKMS}, which covers the Hawking-Unruh type states in Rindler, de Sitter and Anti de Sitter spacetimes.

\section{The thermal energy tensor}

Finally, we discuss a thermal observable which can be interpreted as the thermal energy tensor. As we shall see, in $\mathbf{T}_{\epsilon}$-thermal states (see equation \eqref{Teps} below), one has a local equivalent of the Stefan Boltzmann law. In addition, as has been noted in earlier work on LTE, the usual relation between pressure and energy density for the equilibrium fields leads to an evolution equation for the temperature, which we also briefly discuss here. Moreover, the thermal energy tensor can be used to determine the local rest systems in which a state is in local thermal equilibrium. We note that our definition of the thermal energy tensor differs from the one used previously in the literature \cite{SchlVerc08,Solv10}.

The expectation value of the thermal energy tensor of the scalar field in global equilibrium states in Minkowski spacetime is given by
\begin{displaymath}
E_{ab}(\beta)=\frac{\pi^2}{90\beta^4}(4e_{a}e_{b}+\eta_{ab}),
\end{displaymath}
implying the Stefan-Boltzmann law. Here, $e_{a}$ fixes the ``rest system'' of the medium. An apparent candidate for a thermal observable measuring this quantity is the usual energy momentum tensor $T_{ab}$; but it has been shown in \cite{BuchOjimRoos02} and \cite{Solv10} that this tensor is not a viable thermal observable. This may be seen by decomposing it according to
\begin{equation}\label{Tab}
 T_{ab}=\epsilon_{ab}+\frac{1}{12}(\nabla_{a}\nabla_{b}-g_{ab}\Box)\wick{\phi^2}+\frac{1}{4}g_{ab}\,\wick{\phi P\phi}.
\end{equation}
The first term on the right hand side of \eqref{Tab} is given by
\begin{equation}\label{ThermEnTen}
 \epsilon_{ab}:=-\wick{\phi(\nabla_{a}\nabla_{b}\phi)}+\frac{1}{4}\nabla_{a}\nabla_{b}\wick{\phi^{2}}+\frac{1}{6} R_{ab}\wick{\phi^2}+\frac{1}{4}g_{ab}\,\wick{\phi P\phi}.
\end{equation}
The second term in \eqref{Tab} is a total derivative, so it vanishes in the translationally invariant global equilibrium states in Minkowski spacetime. Finally, $\wick{\phi P\phi}$ is given by some state-independent function $\wick{\phi P\phi} =: U(x)\mathbbm{1}$. This term is responsible for the conformal anomaly, see e.g.\ \cite{More03,HollWald05} for details. Moreover, $U(x)\equiv 0$ in Minkowski spacetime. Therefore, in the global equilibrium states in Minkowski spacetime, both $T_{ab}$ and $\epsilon_{ab}$ have the same expectation values, namely
\begin{displaymath}
 \omega_{\beta}(\epsilon_{{}_{\st{Mink}}ab}(x))=E_{ab}(\beta)=\omega_{\beta}(T_{{}_{\st{Mink}}ab}(x)).
\end{displaymath}
Therefore, $\epsilon_{ab}$ is a candidate for a thermal energy tensor.

We now use the KMS states in the previous section to show that our $\epsilon_{ab}$ is preferable to the previous proposals, which consisted of the first two terms in \eqref{ThermEnTen}. To this end, let us consider the set of thermal observables
\begin{equation}\label{Teps}
 \mathbf{T}_{\epsilon}:=\{\mathbbm{1}\,,\,\wick{\phi^2}\,,\,\epsilon_{ab}\}.
\end{equation}
For the trace of $\epsilon_{ab}$, we find
 \begin{equation}\label{TraceEpsNew}
 \epsilon^{a}{}_{a}=\frac{1}{4}\Box\wick{\phi^2}.
\end{equation}
In contrast to previous prescriptions for the thermal energy tensor in curved spacetime, this relation does not include any explicit curvature terms. From \eqref{TraceEpsNew} and the translation invariance of the global equilibrium states it follows that for a $\mathbf{T}_{\epsilon}(\mathcal{O})$-thermal state $\omega$, one necessarily has
\begin{equation}
 \Box_{{}_{\st{M}}}\omega(\wick{\phi^2}_{{}_{\st{M}}}(x))=0
\end{equation}
in the interior of $\mathcal{O}$. Thus, by equation \eqref{LocalTemp}, $\mathbf{T}_{\epsilon}(\mathcal{O})$-thermality entails an evolution equation for the mean of the square of the local temperature, i.e.\
\begin{equation}\label{EvLocTemp}
 \Box_{{}_{\st{M}}}\overline{T^2(x)}=0.
\end{equation}
As the thermal energy tensor presented here differs from previous ones, the evolution equation \eqref{EvLocTemp} differs from the one given earlier in the literature as well. Our choice, however, is justified by investigating the preceding examples of KMS states. In fact, the Hawking-Unruh type states in de Sitter, Anti de Sitter and Rindler spacetimes satisfy condition \eqref{EvLocTemp} everywhere. Due to additional curvature terms, this would not have been the case with the older definition of the thermal energy tensor appearing in the literature. For $\beta$ other than $2\pi$, the KMS states are generally not $\mathbf{T}_{\epsilon}$-thermal, showing that they deviate quite substantially from an equilibrium situation.

\section{Conclusion}

We have identified particular Wick powers by which a sharp local temperature is assigned to quasi-free LTE states in any spacetime. The remaining freedom consists of the choice of the constant $\alpha$ in the Wick square. Using KMS states in simple stationary spacetimes, we were able to fix $\alpha$ in a physically meaningful way. The fact that the resulting field is conformally covariant provides further support for this choice. We also proposed an improved thermal energy tensor for which the Hawking-Unruh type states in de Sitter and Anti-de Sitter satisfy condition \eqref{EvLocTemp}. Note although we have fixed the renormalization and energy momentum tensor by testing KMS states in some particularly symmetric spacetimes, these choices then carry over to all curved spacetimes.

Most remarkably, we have seen that the KMS property and hence passivity do not entail local thermal equilibrium. In fact, the inverse KMS parameter can in general not be identified with the local temperature in the sense of the zeroth law. But it may be related to it in an interesting way, as is exemplified by equation \eqref{KMSvsLocal}.

\paragraph*{Acknowledgements}

I would like to thank Prof.\ D.\ Buchholz for helpful discussions and for sharing his insight that in curved spacetime, the KMS parameter need not be an inverse temperature in the sense of the zeroth law.

I am grateful to K.\ Fredenhagen and N.\ Pinamonti for bringing reference  \cite{Page82} to my attention and I would also like to thank J.\ Schlemmer and A.\ Stottmeister for their help in the identification of errors in \cite{BuchSchl07} and \cite{Stot09}.

Finally, I would like to thank the Institute for Theoretical Physics in G\"ottingen for hospitality and financial support.

\bibliography{Bibliography}

\end{document}